\documentclass[letterpaper, 10 pt, conference]{ieeeconf}  

\IEEEoverridecommandlockouts                              
\title{\LARGE \bf
Observability of Nonlinear Dynamical Systems over Finite Fields
}

\author{Ramachandran Anantharaman \\ Department of Mathematics \\ Univesity of Namur, Belgium \\ ramachandran.chittur@unamur.be \and Virendra Sule \\ Department of Electrical Engineering \\ Indian Institute of Technology - Bombay, India \\vrs@ee.iitb.ac.in} 

\usepackage{color}
\usepackage{float}
\usepackage{amsmath}
\usepackage{amssymb}
\usepackage{esint}
\usepackage{hyperref}
\usepackage{varioref}

\usepackage{amsthm}

\newtheorem{theorem}{Theorem}
\newtheorem{corollary}{Corollary}
\newtheorem{lemma}{Lemma}
\newtheorem{definition}{Definition}
\newtheorem{remark}{Remark}
\newtheorem{proposition}{Proposition}

\usepackage{algorithm}
\usepackage{algpseudocode}
\usepackage{algorithmicx}
\let\OldStatex\Statex
\renewcommand{\Statex}[1][3]{%
\setlength\@tempdima{\algorithmicindent}%
\OldStatex\hskip\dimexpr#1\@tempdima\relax}

\def\fn{\mathbb{F}^n}
\def\fm{\mathbb{F}^m}
\def\ff{\mathbb{F}}

\def\fn{\mathbb{F}^n}

\def\Vo{\mathcal{F}}

\def\tz{\tilde{z}}

\def\tz{\tilde{z}}
\def\mz{\mathbf{Z}}

\def\matK{\mathbf{K}}
\def\mcb{\mathcal{B}}
\def\beq{\begin{equation}}
\def\eeq{\end{equation}}

\def\phat{\hat{\psi}}
\def\hchi{\hat{\chi}}
\def\hy{\tilde{y}}

\def\mco{\mathcal{O}}

\def\Ffn{\mathcal{F}({\mathbb{F}^n})}

\def\funf{\mathcal{F}(\mathbb{F})}
\def\hx{\tilde{x}}
\def\hF{\tilde{F}}
\def\hg{\tilde{g}}
\def\ST{h}

\def\Koop{\mathcal{K}}

\begin{document}
\title{Observability of Nonlinear Dynamical Systems over Finite Fields}
\author{Ramachandran Ananthraman and Virendra Sule 
\thanks{R Ananthraman is currently with the Department of Mathematics and Namur Institute for Complex Systems (naXys), University of Namur, Belgium. {\tt ramachandran.chittur@unamur.be}}
\thanks{V Sule is a retired Professor at Department of Electrical Engineering, Indian Institute of Technology-Bombay, India. {\tt vrs@ee.iitb.ac.in}}
}
\maketitle
\begin{abstract}
This paper discusses the observability of nonlinear Dynamical Systems over Finite Fields (DSFF) through the Koopman operator framework. In this work, given a nonlinear DSFF, we construct a linear system of the smallest dimension, called the Linear Output Realization (LOR), which can generate all the output sequences of the original nonlinear system through proper choices of initial conditions (of the associated LOR). We provide necessary and sufficient conditions for the observability of a nonlinear system and establish that the maximum number of outputs sufficient for computing the initial condition is precisely equal to the dimension of the LOR. Further, when the system is not known to be observable, we provide necessary and sufficient conditions for the unique reconstruction of initial conditions for specific output sequences.  
\end{abstract}
\section*{Notations}
$\ff$ denotes a finite field and $\ff_q$ denotes the finite field of $q$ elements, where $q = p^m$, for some prime $p$. $\funf$ denotes the vector space of $\ff$-valued functions over $\ff$. $\ff^n$ denotes the Cartesian space of $n$-tuples over $\ff$ and $\Ffn$ denotes the vector space of $\ff$-valued functions over $\ff^n$. The coordinate functions $\chi_i \in \Ffn$ are defined by their values $\chi_i(x) = x_i$ for $x \in \ff^n$ and $i$-th co-ordinate $x_i \in \ff^n$. 

\section{Introduction}
Dynamical Systems over Finite Fields (DSFF) are a special class of dynamics where the state space of the system is considered to be a vector space over a finite field $\ff$. Such systems arise in models of Boolean Networks, Communication and cryptographic protocols, Genetic and Biological Networks, and Finite State Automata, to name a few. This class of systems has been studied in detail from a dynamical systems point of view \cite{Gill1, Gill2}. When the dynamics is nonlinear, the problems of computing the solutions, deciding controllability, observability, and optimal control are inherently hard as they require solving nonlinear algebraic equations over finite fields \cite{RamSule}. In this work, we address the problem of observability of nonlinear DSFF. Observability for dynamics over Boolean Dynamics has been previously addressed in the works of \cite{DCheng_2009Obs, ValcherObs2013,Laschov2013, Weiss2018} and mainly depends on the semi-tensor product based approaches towards Boolean Dynamics developed in \cite{DCheng_STP2010}. These works, though solve for the observability of nonlinear systems through a linear algebraic framework, require computations over exponential size matrices, leading to poor scalability with respect to the dimension of the system. The work \cite{RamSule} provides an alternative linear theory for nonlinear DSFF through the Koopman operator framework without necessarily requiring exponential sized linear systems. In this paper, we focus on problem of observability in such nonlinear systems by defining a minimal linear realization (called as the Linear Output Realization) of a nonlinear autonomous system using the Koopman operator. With this LOR, we aim to give necessary and sufficient conditions for the observability of the dynamical system and guarantee upper bounds on the number of outputs required for the unique reconstruction of the initial condition from the output. 

The paper is organized as follows. The following section establishes the preliminaries and discusses the problem formulation. Section \ref{s3} discusses the construction of the LOR and establishes invariance of the dimension of the LOR with respect to a nonsingular (and possibly) nonlinear state transformation of the system. Section \ref{s4} gives conditions for observability of the nonlinear DSFF through the LOR and an upper bound on the number of outputs sufficient for retrieving the initial conditions. 
\section{Preliminaries}
\label{s1}
Consider a nonlinear DSFF represented as
\begin{align}
    \label{eq:DSFF}
    \begin{aligned}
    x(k+1) = F(x(k)) \\
    z(k) = g(x(k)),
    \end{aligned}
\end{align}
where $x(k) \in \fn$, $z(k) \in \fm$, $F(x): \fn \to \fn$ and $g(x): \fn \to \fm$ are the state, output, state transition map and output map respectively. 

\begin{definition}[Observable states]
A state $x_0 \in \ff^n$ is said to be observable if there exists a positive integer $K$ such that $x_0$ is the unique initial condition of the output sequence $z(k)$ it generates, for $k = 0,\dots,K-1$. The smallest such $K$, denoted as $K_{x_0}$ is defined as the observability index of the initial condition $x_0$.  
\end{definition}
\begin{definition}[Observability, Observability index and Observability window of a system]
A system is said to be observable if every initial condition $x_0 \in \ff^n$ is observable. According to Definition 1, an observable system has a unique $K_{x_0}$ for every initial condition $x_0$. The maximum $K = max\ K_{x_0}$ over all $x_0 \in \ff^{n}$, is called the observability index of the system, and the interval $[0,K-1]$ is defined as the observability window of the system (\ref{eq:DSFF}).
\end{definition}
\begin{proposition}
For an observable time-invariant DSFF (\ref{eq:DSFF}), $K$ is always finite.
\end{proposition}
\begin{proof}
For a DSFF as in equation (\ref{eq:DSFF}) over $\ff_q^n$, the total number of points in the state space is finite and equal to $q^n$, and the outputs $y \in \ff_q^m$, and the output space contains $q^m$ distinct points. For a shift-invariant DSFF, any initial condition $x_0$ will either be visited periodically with period $\leq q^n$ or will settle down into a periodic orbit (with period $\geq 1$) under the dynamics of (\ref{eq:DSFF}) and hence ultimately periodic \cite{RamSule}. The corresponding output sequence is ultimately periodic, too. Generally, any output sequence contains two parts: the pre-periodic and periodic components. As the output space is finite, both the pre-periodic component and the period of the periodic component are finite. Given that the dynamics is shift-invariant, the initial conditions from the output has to be essentially computed from the outputs $z(k)$, $k = 0,1,\dots,K$, where $K$ is time instant corresponding to the first period of the output orbit after the pre-periodic component. Hence, the number of outputs required to compute the initial condition is finite. 
\end{proof}
 

When the system (\ref{eq:DSFF}) is linear with state transition matrix $A$ and output matrix $C$, a necessary and sufficient condition for the system to be observable is that the following matrix $\mco$ (known as the \emph{Observability matrix}) 
\begin{align}
\label{eq:Obsmat}
    \mco = \begin{bmatrix}C \\ CA \\ \vdots \\ CA^{n-1}\end{bmatrix}
\end{align}
is full rank and the corresponding initial condition $x_0$ which generates the given sequence is the solution of the following system of linear equations
\begin{align}
\label{eq:ObsEqLin}
\mco x_0 =  \begin{bmatrix} z(0) \\ z(1) \\ \vdots \\ z(n-1) \end{bmatrix}.
\end{align}
In linear systems, as a consequence of Cayley Hamilton Theorem, at most $n$ output instances are required to compute the corresponding initial condition $x_0$. In general, for a multi-output linear system, if the system is observable, the observability window can be computed as $[0,n_o - 1]$, where $n_o$ is the observability index\footnote{The observability index of a linear system is the minimal $k$ such that $CA^{k}$ is linearly dependent on the previous iterates $CA^{i}$.} of the system. Also, it can be seen that for linear systems, the observability of one state guarantees the observability of the entire state space (and hence the system itself). 
 \begin{proposition}
 In a linear system, the observability of one particular initial condition is equivalent to the observability of the entire system. 
\end{proposition}
 \begin{proof}
 Given that a specific initial condition $x_0$ is observable from its output sequence $z(0),z(1),z(2),\dots,z(n-1)$. Then $x_0$ is the unique solution to the system of equation (\ref{eq:ObsEqLin}. If the system is unobservable, $\mco$ is singular, and a nontrivial kernel $\mathcal{N}(\mco)$ exists. Let $w \in \mathcal{N}(\mco)$, then we know that 
 \[
 \mco (x_0 + w) = \mco x_0,
 \]
 and $x_0+w$ is also a solution of (\ref{eq:ObsEqLin}) which contradicts the observability of $x_0$. 
\end{proof}

We denote by $\mz$, the output sequence $z(k)$ of arbitrary length, and by $\mz_K$ an output sequence of length $K$ at time instants $0,1,\dots,K-1$.
\begin{definition}[Observability of an output sequence]
An output sequence $\mz_K$ of length $K$ is said to be observable if there exists a unique initial condition $x_0$ which generates the sequence.     
\end{definition}
This notion of observability is defined from the context of a specific output sequence of length $K$ rather than from the context of the initial condition. It can be seen that both these definitions are interchangeable: When an initial condition is observable, its corresponding output sequence $\mz_{K_{x_0}}$ is also observable, where $K_{x_0}$ is the observability index of $x_0$. If the system is observable, it can be seen that every output sequence $\mz_K$ is observable, where $K$ is the observability index of the system. The idea of looking at the observability of a system from the output sequence is important because, in many practical applications, it is necessary to ascertain the uniqueness of the initial condition for a given output sequence $\mz_K$ of length $K$. 


\subsection{Challenges in Observability of a nonlinear DSFF}
For a nonlinear DSFF, determining the observability of an initial condition $x_0$ involves establishing the existence of a positive integer $K$ such that the following system of algebraic equations (\ref{eq:ObsEq}) has $x_0$ as the unique solution. 
\begin{align}
\label{eq:ObsEq}
\begin{aligned}
    z(0) &= g(x_0) \\
    z(1) &= g(F(x_0) \\
    & \ \ \vdots \\
    z(k-1) &= g(F^{(K-1)}(x_0)),
\end{aligned}
\end{align}
where $F^{(j)}$ is the $j^{th}$ composition of $F$ with itself 
\[
F^{(k)} = \underbrace{F \circ F \circ \dots \circ F}_{j\ \text{times}}.
\]
Computing $x_0$ and establishing its uniqueness for a given output sequence as in (\ref{eq:ObsEq}) is called the \emph{observability problem}. Computing the solution of the nonlinear system of equations (\ref{eq:ObsEq}) over a finite field $\ff$ is an inherently hard computational problem. We highlight the challenges in determining the observability of a nonlinear DSFF.
\begin{enumerate}
    \item Solving for multivariate polynomial equations (\ref{eq:ObsEq}) over finite fields is an $NP$-hard problem \cite{LB_Hybrid2009,JDing}. No known efficient algorithm exists for general finite fields. Computer algebraic approaches based on Gr\"{o}bner basis techniques can be used for computing such solutions. However, these techniques are not known to be efficient in scaling up with the dimension of the system.
    \item For a given output sequence $\mz_K$, it is guaranteed that there exists atleast one initial condtion $x_0$ as a solution of (\ref{eq:ObsEq}). It is, however, pertinent to ascertain the uniqueness of the solution $x_0$ to (\ref{eq:ObsEq}) because the computation of a solution $x_0$ does not guarantee the observability of that output sequence.
    \item The definition of observability of a dynamical system implicitly involves the length of the observability window $[0,K-1]$, which guarantees the uniqueness of $x_0$. This window is well-defined for linear systems and equal to $[0,n_o-1]$, where $n_o$ is the observability index. There are no specific results on the upper bound of $K$ for general nonlinear dynamics. A result on the estimate of $K$ is essential for deciding the observability of output and computation of the initial conditions. 
    
    \item Unlike in linear systems, where observability of a single initial condition (or a single output sequence) guarantees the observability of all initial conditions (or of all output sequences), nonlinear systems may have certain initial conditions (or output sequences) to be observable. In contrast, other initial conditions (or output sequences) may not be observable. So, the uniqueness of the solution $x_0$ for a particular output sequence $\mz$ does not guarantee the uniqueness of initial conditions of all output sequences, making the characterization of the observability of the entire system difficult.
    
\end{enumerate}

With the above motivations for the observability of nonlinear DSFF, this paper takes a fresh look at the observability problem through the Koopman Operator framework. 

\subsection{Koopman Operator Framework for DSFF}
Consider the space of functions $\Vo$ of scalar-valued functions $\psi : \fn \to \ff$. This space is of finite dimension (and equal to $q^n$) when $\ff$ is the finite field $\ff_q$. The \emph{Koopman operator} $\Koop$ \cite{Koopman} associated with the dynamics (\ref{eq:DSFF}) is defined as 
\begin{align}
\label{eq:KoopmanOp}
    \Koop \psi = \psi \circ F.
\end{align}
This map $\Koop$ is linear over $\Vo$. Starting from a function $\psi \in \Vo$, one can define the sequence 
\[
\psi, \Koop \psi := \psi(F), (\Koop)^2 \psi := \psi(F^{(2)}),\dots
\]
The iterates of the above sequence can be visualized as the evolution of the following dynamical system 
\begin{align}
\label{eq:KLS}
\psi_{k+1} = \Koop \psi_{k},
\end{align}
starting from $\psi_0 = \psi$. The system (\ref{eq:KLS}) is over $\Vo$ and is $\ff$-linear and defined as the \emph{Koopman linear system} (KLS) for the dynamical system (\ref{eq:DSFF}). The work \cite{RamSule} establishes the relationships between the nonlinear DSFF and the KLS. 

This paper focuses on the solution to the observability problem (\ref{eq:ObsEq}) and develops a necessary and sufficient condition for the observability of the nonlinear DSFF through the Koopman operator framework by restricting it to a specific invariant subspace of $\Vo$.

\section{Linear Output Realization of a DSFF through Koopman Operator}
\label{s3}
In this section, we seek to construct a linear DSFF of smallest dimension through the Koopman operator, which can reproduce each output sequence of the original nonlinear DSFF. Consider the output map of (\ref{eq:DSFF})  to be $g= [g_1,g_2,\dots,g_m]^T$, where each $g_i \in \Vo$ is called the \emph{output function}. Starting from $g_1$, one can construct the sequence 
\[
g_1,\Koop g_1,\Koop^2 g_1 \dots,
\]
Owing to the finite dimension of the space $\Vo$, there exists a smallest $l_1 > 0$ such that $\Koop^{l_1}g_1$ is linearly dependent on its previous iterates. The cyclic subspace of $\Vo$ generated by $g_1$ and having the basis $\{g_1, \Koop g_1,\dots, \Koop^{l_1-1}g_1 \}$ is $\Koop$-invariant and the smallest such invariant subspace containing the first output function $g_1$. 

This space is further expanded by the addition of subsequent output functions $g_i$, and finally, an invariant subspace of $\Vo$ containing all the output functions $g_i$ is computed. The procedure to construct the smallest $\Koop$-invariant subspace $W$ containing all the output functions is given in Algorithm \ref{alg:LORConst}.

\begin{algorithm}[h]
\begin{algorithmic}[1]
\caption{Construction of $W$ - $\Koop$-invariant subspace spanning $g_i$}
\label{alg:LORConst}
\Procedure{$\Koop$-Invariant subspace}{$W$}
\State \textbf{Outputs}: 
\begin{itemize}
    \item[] $W$ - the invariant subspace which span the output functions and $\Koop$-invariant 
    \item[] $\mathcal{B}$ - the basis for the subspace $W$
\end{itemize}
\State Compute the cyclic subspace generated by $g_1$ 
\[ 
Z(g_1, \Koop) = \langle g_1,  \Koop g_1,\dots, (\Koop)^{l_1-1} g_1 \rangle
\]
\label{A1:p1}
\State Set of basis functions $\mathcal{B} = \{g_1,\Koop g_1,\dots,\Koop^{l_1-1}g_1 \}$
\If{$g_2,g_3,\dots,g_m \in \mbox{Span}\{\mathcal{B}\}$} 
\State $W \gets \mbox{Span}\{\mathcal{B}\}$
\State \textbf{halt}
\Else
\State Find the smallest $i$ such that $g_i \notin \mbox{span}\{\mathcal{B}\}$
\State Compute the smallest $l_i$ such that 
\[
\Koop^{l_i} g_i \in \mbox{Span}\{ \mathcal{B} \cup \langle g_i,\Koop g_i,\dots, \Koop^{l_i-1}g_i \rangle\} 
\]
\label{A1:p2}
\State $V_i = \{g_i,\Koop g_i,\dots, \Koop^{l_i-1} g_i \}$
\State Append the set $V_i$ to $\mathcal{B}$
\State \textbf{go to} 5
\EndIf
\EndProcedure
\end{algorithmic}
\end{algorithm}

Choosing $\mcb = \{\psi_1,\dots,\psi_N\}$ to be a basis of $W$, define $\hat{\psi}$ to be the vector of basis functions 
\begin{align}
\label{eq:phat}
\phat = \begin{bmatrix} \psi_1 \\ \vdots \\ \psi_N\end{bmatrix}
\end{align}
and since $W$ is $\Koop$-invariant,
\[
\Koop \phat = \begin{bmatrix} \Koop \psi_1  \\ \vdots \\ \Koop \psi_N\end{bmatrix} =: \matK \phat,
\]
where $\matK \in \ff^{N \times N}$ is the matrix representation restriction of $\Koop|W$ with the basis $\mcb$. Since every $g_i \in W$, there exists an unique matrix $\Gamma \in \ff^{m \times N}$ such that 
\begin{align}
\label{eq:Gamma_map}
g= \Gamma \phat.
\end{align}
The following linear system over $\ff^N$ with output over $\ff^m$ is defined as the \textit{linear output realization} (abbreviated as LOR) of the system (\ref{eq:DSFF}).
\begin{align}
    \label{eq:LOR}
    \begin{aligned}
    y(k+1) &= \matK y(k) \\
    \tz(k) &= \Gamma y(k),
    \end{aligned}
\end{align}
where $y(k) \in \ff^N$ and $\tz(k) \in \ff^m$ are the internal states and the output of the LOR, respectively. 

The following proposition establishes that the LOR reproduces all the output sequences of the original nonlinear DSFF. 
\begin{proposition}
\label{prop:OE}
Given a DSFF as in (\ref{eq:DSFF}) and its corresponding LOR (\ref{eq:LOR}), the outputs $z(k)$ of the DSFF starting from an initial state $x_0$ are identical to the outputs $\tz(k)$ of the LOR when the LOR is initiated as $y_0 = \phat (x_0)$.
\end{proposition}
The proof is similar to Theorem 2 of \cite{RamSule} and omitted for brevity. 

The main difference between this work and \cite{RamSule} is in the construction of the invariant subspace. In this work, we restrict the Koopman operator to the smallest invariant subspace containing the output functions. In contrast, in the previous work, we had restricted the Koopman operator to the smallest invariant subspace containing both the coordinate and output functions, leading to different linear operators and linear systems. The restriction proposed in this work can be viewed as the construction of a (minimal) linear realization of the output of a DSFF, while the restriction proposed in the previous work can be viewed as the construction of a linear representation for a nonlinear system. 

\begin{lemma}
\label{lem:LOR-obs}
The LOR (\ref{eq:LOR}) of the DSFF (\ref{eq:DSFF}) is observable.
\end{lemma}

\begin{proof}
The chosen basis of $W$ in Algorithm \ref{alg:LORConst} helps in proving the observability of the LOR. The basis $\mcb$ contains $g_i$ for some $i \subset \{1,2,\dots,m\}$ and their iterates under $\Koop$. Also, since the dimension of the subspace $W$ is $N$, the upper bound $i$ on $(\Koop)^i g_j $ as a basis function of $W$ is $N-1$ (and this can happen only if the iterates of $g_1$ span the entire subspace $W$). The observability matrix of the LOR is given by
\[
\mco = \begin{bmatrix} \Gamma \\ \Gamma \matK \\ \vdots \\ \Gamma \matK^{N-1}\end{bmatrix}.
\]
We prove that the rank of this matrix is $N$, which proves the observability claim. The action of $\mco$ on $\phat$ is given as 
\begin{align*}
    \mco \phat = \begin{bmatrix} \Gamma \phat \\ \Gamma \matK \phat \\ \vdots \\ \Gamma \matK^{N-1} \phat \end{bmatrix}.
\end{align*}
Each of the submatrix $\Gamma \matK^{i} \phat$, $i = 0,1,\dots, N-1 $ can be written as 
\begin{align}
\label{eq:op_rec}
\begin{aligned}
\Gamma \matK^{i}\phat(x) &= \Gamma (\Koop)^{i} \phat(x) \\
 &= \Gamma \phat (F^{(i)}(x)) \\
 &= g(F^{(i)}(x)).
 \end{aligned}
\end{align}
This proves that all the iterates of the $g_i$ under $\Koop$ (or equivalently the functions $g_i(F^{(i)}$), $i = 0,1,\dots,N-1$ are in the rows of the observability matrix $\mco$. Hence, all the basis functions $\psi_i$ of $W$ form a subset of the rows of the observability matrix $\mco$. This means that $\mco$ has $N$ independent rows, which proves the claim. \qed
\end{proof}
The unique $y_0$ for a given output sequence $\mz$ is computed as the solution of the following equation
\begin{align}
\label{eq:phat1}
\begin{bmatrix}
\tz(0) \\ \tz(1) \\ \vdots \\ \tz(N-1) 
\end{bmatrix} = \mco y_0. 
\end{align}

\begin{remark}
The LOR for a DSFF is always observable, irrespective of the observability of the original nonlinear DSFF. The observability of the original DSFF needs a unique computation of $x_0$, while the observability of the LOR requires only the unique computation of $\phat$ at $x_0$, the initial state of the LOR. This result emphasises that the lack of observability of the original nonlinear DSFF does not carry forward as a property of the LOR. 
\end{remark}
\begin{remark}
The LOR for a nonlinear DSFF can be visualized as the smallest dimensional linear system (smallest because it is constructed from the smallest invariant subspace of $\Vo$ containing all the output functions), which can recreate all the output sequences of the original system through proper choices of initial condition $y_0$. For a linear system, this would correspond to the observable part in the Kalman canonical realization. \end{remark}

The following result shows that subspace $W$ is invariant under a nonsingular state transformation of the DSFF. In particular, we see that for a chosen basis of $W$, two DSFFs that are related by a nonsingular state transformation result in the same LOR.   
\begin{theorem}
\label{lem:NS_Transform}
Given two nonlinear DSFFs $S_1$ and $S_2$ related by a nonsingular (possibly nonlinear) state transformation, their corresponding LORs $L_1$ and $L_2$ are of the same dimension and related by a linear state transformation.
\end{theorem}

The proof is provided in the appendix. With this construction of the LOR that generates all the output sequences of a nonlinear system by a suitable choice of initial condition, we use this LOR to analyze the observability of the nonlinear DSFF. 


\section{Observability of the DSFF through LOR}
\label{s4}
Given $\Vo$ over $\ff^n$, let $\chi_i(x) \in \Vo $ be the $i^{th}$ coordinate function defined as
\[
 \chi_i(x) = x_i.
\]
The state of the system at a point $x_0 \in \ff^n$ is the evaluation of the $n$-tuple of coordinate functions $\hchi = [\chi_1,\dots,\chi_n]^T$ at $x_0$. The unique reconstruction of the initial state $x_0$ of the DSFF from the output sequence is equivalent to the unique reconstruction of the $\hchi$ at $x_0$ from the output sequence. The following result characterizes the observability of the nonlinear DSFF through the LOR.

\begin{theorem}
\label{thm:obs_phat}
The DSFF is observable iff the map $\phat: \ff^n \to \ff^N$ defined in (\ref{eq:phat}) is an injective map of $\ff^n$ into $\ff^N$. 
\end{theorem}
\begin{proof}
(Sufficiency): Given that the map $\phat$ is an injection, for every $\phat(x_0) \in \ff^N$, there exists an unique $x_0 \in \ff^n$. Consider any output sequence $z(k)$ of the DSFF. For the assignment $\tz(k) = z(k)$ of the LOR and using the fact that the LOR is observable, there exists a unique $y_0$ satisfying $y_0 = \phat(x_0)$. Under the assumption that $\phat$ is $1-1$, there exists a unique $x_0$ corresponding to the output sequence $z(k)$ of the DSFF, which proves that the DSFF is observable. 

(Necessity): Given that the DSFF is observable, for every output sequence $z(k)$, there exists a unique $x_0$ which generates that sequence. When LOR is initiated with $\phat(x_0)$, the output sequence $\tz(k)$ of the LOR is identical to the output sequence $z(k)$ of the DSFF. Since LOR is observable, the initial condition $\phat(x_0)$ is unique for a given output sequence. Hence, the map $\phat(x_0)$ is $1-1$ as each $x_0$ generates a unique output sequence $z(k)$. 
\end{proof}


\begin{corollary}
\label{cor:DSFF-obs_suf}
The DSFF is observable if all the coordinate functions are in $W$
\end{corollary}
\begin{proof}
Given all the coordinate functions are in $W$, there exists a unique matrix $C \in \ff^{n \times N}$ such that 
\begin{align}
\label{eq:statemap}
\hchi = C \phat.
\end{align}
From LOR and Proposition \ref{prop:OE}, we know that if LOR is initialized with $y_0 = \phat(x_0)$, then $z(k) = \tz(k)$. To compute this $\phat(x_0)$, we solve for $y_0$ in (\ref{eq:phat1}). Since LOR is observable, a unique $\phat(x_0)$ exists for each output sequence $z(k)$. This unique $\phat(x_0)$ gives the unique $x_0$ through equation (\ref{eq:statemap}) as
\[
x_0 = \hchi(x_0) = C \phat(x_0),
\]
which is the unique initial condition corresponding to the output $z(k)$. 
\end{proof}
 When the system is linear, the above corollary is both necessary and sufficient, as the observability matrix is full rank, thereby enabling the recovery of initial condition from the outputs as a solution to a linear system of equations (\ref{eq:ObsEqLin}). 
The following result establishes the connection between the number of outputs required for the unique reconstruction of the initial condition and the dimension of the LOR. Such a result is important for establishing computational bounds on the observability problem of nonlinear DSFF.
\begin{theorem}
Given an observable DSFF (\ref{eq:DSFF}) and its corresponding LOR (\ref{eq:LOR}), then the maximum number of outputs required for the unique reconstruction of the initial condition is $N$.
\end{theorem}
\begin{proof}
As the DSFF is observable, then there exists a $K > 0$ such that the following set of equations have a unique solution $x_0$
\small{\begin{align}
\label{eq:ObsK}
\begin{aligned}
z(0) &= g(x_0) \\ 
&\ \ \vdots  \\
z(N) &= g(F^{(N)}(x_0)) \\
&\ \ \vdots \\
z(K-1) &=g(F^{(K-1)}(x_0)). \\
\end{aligned}
\end{align}}
\normalsize{}
 We claim that the last $K-N$ equations can be eliminated by linear combinations of the first $N$ equations. From equation (\ref{eq:op_rec}) we see that 
 \[
 g(F^{(N)}(x)) = \Gamma \matK^{N}\phat(x).
 \]
 As a consequence of Cayley Hamilton Theorem, we can see that 
\small{ \begin{align*}
     g(F^{(N)}(x)) &= \Gamma \matK^{N} \phat(x) =  \Gamma \big( \sum_{i = 1}^{N} \alpha_i \matK^{N-i} \big)  \phat(x) \\
     &= \sum_{i=1}^{N} \alpha_i \Gamma \matK^{N-i} \phat(x) =  \sum_{i=1}^{N} \alpha_i g(F^{(N-i)}(x),
 \end{align*}}
 \normalsize{}
 where $\alpha_i$ are the coefficients corresponding to the characteristic polynomial of $\matK$ given as 
 \small{\[
 p(X) = X^N - \sum_{i=1}^N \alpha_i X^{N-i}.
 \]}
 \normalsize{}
By extension, we see that all the equations corresponding to $g(F^{(N+k)}(x))$, $k>0$ are linearly dependent on the first $N$ equations and hence the system of equations (\ref{eq:ObsK}) reduces to 
\begin{align}
\label{eq:ObsEqns}
\begin{aligned}
z(0) &= g(x_0) \\ 
&\ \ \vdots  \\
z(N-1) &= g(F^{(N-1)}(x_0)).
\end{aligned}
\end{align}
From the observability of the DSFF, we know that the above equation has a unique initial condition $x_0$ for every output sequence of the system. 
\end{proof}

As a consequence of the above theorem, we have the following result about the observability of a specific output sequence even when the system is unobservable.
\begin{corollary}
Any initial condition $x_0$ is observable iff it can be uniquely computed from the output sequence $Z_{N-1}$.
\end{corollary}
An important consequence of the above two results is that when the uniqueness of the initial condition from the output is not ascertained with $N$ outputs, it will never be ascertained with any additional outputs. Such a result is important as this upper bound helps to decisively comment about the observability of a nonlinear system (or a particular output sequence).

\section{Conclusion}
In this paper, we developed a linear approach for the observability of nonlinear DSFF using the Koopman operator. Through this framework, we defined the LOR, the linear output realization for a nonlinear system, which can be viewed as the smallest dimensional linear system that can reproduce all the outputs of the original nonlinear system. The dimension of this linear system is defined through the construction of an invariant subspace of the Koopman operator. It is further established that the maximum number of outputs required to determine the initial condition from the corresponding output sequence uniquely equals the dimension of the invariant subspace. This work is a preliminary investigation of the observability of nonlinear DSFF, and important future directions from this work will include constructing efficient algorithms to determine the injectivity of the map $\phat$ for specific classes of systems or specific fields $\ff$. In particular, it would be interesting to explore this problem for Boolean Dynamics, where results from Boolean theory can be used. Further, we can explore construction of observers for nonlinear DSFF using this framework. 
\appendix
\subsection*{Proof of Theorem \ref{lem:NS_Transform}:}
Consider two non-linear DSFFs $S1, S2$ defined as 
\begin{align*}
\begin{aligned}
    x(k+1) &= F(x(k)) \\ 
    z(k) &= g(x(k)) 
\end{aligned} 
\hspace{.2in}
\begin{aligned}
    \hx(k+1) &= \hF(\hx(k)) \\ 
    z(k) &= \hg(\hx(k)) 
\end{aligned}
\end{align*}
related by a non-singular state transformation $\hx = \ST(x)$. Let their corresponding LOR be $L1,L2$ and defined through the following equations
\begin{align*}
\begin{aligned}
    y(k+1) &= \matK y(k) \\
    z(k) &= \Gamma y(k) 
    \end{aligned}
    \hspace{.2in}
    \begin{aligned}
        \hy(k+1) &= \tilde{\matK} \hy(k) \\
        z(k) &= \tilde{\Gamma} \hy(k).
    \end{aligned} 
\end{align*}
constructed as per Algorithm \ref{alg:LORConst} with invariant subspace $W$ and $\tilde{W}$ with basis functions $\psi_j$ and $\tilde{\psi}_j$.

First we show that $\hF \circ \ST = \ST \circ F$
\begin{align}
\label{eq:temp1}
\begin{aligned}
\hF \circ \ST(x(k))  &= \hF(\hx(k)) = \hx(k+1) = \ST(x(k+1)) \\
& = \ST(F(x(k)) = \ST \circ F(x(k)) \quad \forall\ x(k)
\end{aligned}
\end{align}
Similarly, we can prove that $\hg \circ \ST = g$. Further,
\begin{align}
\label{eq:Fg-equivalence}
\begin{aligned}
g(F^{(i)}(x))  = (\hg \circ \ST) \circ F^{(i)}(x)  = \hg \circ (\ST \circ F^{(i)}) (x) \\
= \hg  \circ (\hF^{(i)} \circ \ST) (x)  = \hg(\hF^{(i)} (\hx))
\end{aligned}
\end{align}
for all $x$ and $\hx = \ST(x)$. 

In the construction of the invariant subspace $W$, we compute functions $g_i (F^{(j)}(x))$ as defined in statements 3 or 10 of Algorithm \ref{alg:LORConst}. It can be seen from the equation (\ref{eq:Fg-equivalence}) that for all $x \in \ff^n$, and $\hx$ such that $\hx = \ST(x)$ 
\begin{align}
\label{eq:temp2}
g_i (F^{(j)}(x)) = \hg_i (\hF^{(j)}(\hx)) \quad \forall i = 1,\dots,m \quad j \geq 0 
\end{align}
This establishes that $l_1$ and $l_i$ as computed in statements 3 or 10 of Algorithm \ref{alg:LORConst} are the same irrespective of the system being $S1$ or $S2$. This proves that the dimension of the invariant subspace is the same. Also since the basis functions $\psi_k$ as defined in Algorithm \ref{alg:LORConst} consists of functions only of the form $g_i (F^{(j)}(x))$, from (\ref{eq:temp2}) there is an 1-1 correspondence between the basis functions as follows
\begin{align}
\label{eq:temp3}
\begin{aligned}
\tilde{\psi}_k(\hx) &= (\tilde{\psi}_k \circ h)(x) = \tilde{g}_i (\tilde{F}^{(j)}(\tilde{x})) \\ 
&= (g_i \circ h^{-1})\circ (h \circ F^{(j)} \circ h^{-1}) \circ (h^{-1} \circ x) \\
& = g_i(F^{(j)}(x)) = \psi_k(x)
\end{aligned}
\end{align}
Further for this chosen $\hat{\psi}(x)$ and $\widehat{\tilde{\psi}}(x)$, 
\small{\begin{align*}
\Koop \phat = \begin{bmatrix} \psi_1 \circ F(x) \\ \vdots \\ \psi_N \circ F(x) \end{bmatrix} &=  \begin{bmatrix} \psi_1 \circ \ST^{-1} \circ \ST \circ F \circ \ST^{-1} \circ \ST (x) \\ \vdots \\ \psi_N \circ \ST^{-1} \circ \ST \circ F \circ \ST^{-1} \circ \ST (x) \end{bmatrix}  \\
&= \begin{bmatrix} \tilde{\psi}_1 \circ \hF (\hx) \\ \vdots \\ \tilde{\psi}_N \circ \hF(\hx)  \end{bmatrix}  = \tilde{\Koop} \widehat{\tilde{\psi}} 
\end{align*}}
\normalsize

The above equation proves that the action of $\Koop$ on the basis function $\phat$ is the same as the action of $\tilde{\Koop}$ on the basis functions $\widehat{\tilde{\psi}}$. This along with the $1-1$ correspondence between the basis functions (\ref{eq:temp3}) gives that the matrices $\matK$ and $\tilde{\matK}$ are the same with the chosen basis $\phat$ and $\widehat{\tilde{\psi}}$. Similarly, it can be proven that with the matrices $\Gamma$ and $\tilde{\Gamma}$ are the same matrices. This proves that the LOR for the systems $S1$ and $S2$ are the same for this particular choice of basis functions $\mcb$ as defined in Algorithm \ref{alg:LORConst}. In principle, one can choose different basis for the LORs and then there exists a linear state transformation between the LORs $L1$ and $L2$. \hfill \qed





\end{document}